%% file: main.tex
\documentclass[runningheads]{llncs}

\usepackage{microtype}
\usepackage[T1]{fontenc}
\usepackage{amssymb}

\usepackage{amsmath}

\usepackage{amsthm}

\usepackage{mathtools}
\usepackage{booktabs}
\usepackage{xspace}
\usepackage{tikz}
\usepackage{hyperref}
\usepackage[inline]{enumitem}

\renewcommand{\phi}{\varphi}
\newcommand{\eval}{V}
\newcommand{\outcome}{out}
\newcommand{\ap}{\mathrm{AP}}
\newcommand{\boolfour}{\mathbb{B}_4}

\newcommand{\atl}{ATL\xspace}
\newcommand{\atlstar}{ATL$^*$\xspace}
\newcommand{\atlstarf}{ATL$^*$[$\mathcal F$]\xspace}
\newcommand{\ratlstar}{rATL$^*$\xspace}
\newcommand{\atlf}{ATL[$\mathcal F$]\xspace}
\newcommand{\ratl}{rATL\xspace}
\newcommand{\ctl}{CTL\xspace}
\newcommand{\rctl}{rCTL\xspace}
\newcommand{\ctlstar}{CTL$^*$\xspace}
\newcommand{\rctlstar}{rCTL$^*$\xspace}
\newcommand{\ltl}{LTL\xspace}
\newcommand{\rltl}{rLTL\xspace}
\newcommand{\ltlf}{LTL[$\mathcal F$]\xspace}

\newcommand{\set}[1]{\{#1\}}
\newcommand{\pow}[1]{2^{#1}}
\newcommand{\nats}{\mathbb{N}}
\newcommand{\element}[2]{#1[#2]}
\newcommand{\ptime}{\textsc{PTime}\xspace}
\newcommand{\exptime}{\textsc{ExpTime}\xspace}
\newcommand{\twoexptime}{\textsc{2ExpTime}\xspace}

\newcommand{\states}{St}
\newcommand{\ag}{{Ag}}
\newcommand{\ac}{{Ac}}
\newcommand{\tr}{\delta}
\newcommand{\lab}{\ell}
\newcommand{\struct}{\mathcal{S}}
\newcommand{\pth}{\pi}
\newcommand{\size}[1]{|#1|}
\newcommand{\strat}{f}
\newcommand{\stratset}{F}
\newcommand{\win}{\mathrm{Win}}
\newcommand{\game}{\mathcal{G}}
\newcommand{\trans}{\mathcal{K}}

\renewcommand{\next}{\mathrm{Next}}
\newcommand{\reach}{\mathrm{Reach}}
\newcommand{\safe}{\mathrm{Safety}}
\newcommand{\buchi}{\mathrm{B\ddot{u}chi}}
\newcommand{\cobuchi}{\mathrm{coB\ddot{u}chi}}

\newcommand{\nrsbc}{NRSBC\xspace}

\newcommand{\tval}{t}
\newcommand{\satset}{\mathrm{Sat}}

\renewcommand{\vec}{v}
\newcommand{\vecs}{AV}

\tikzset{robust/.style={line width=.16ex,line join=round}}

\let\Box\relax
\DeclareMathOperator{\Box}{%
	\text{%
		\tikz[baseline]{%
    			\draw[robust] (0ex,-.1ex) -- (0ex, 1.4ex) -- (1.5ex, 1.4ex) -- (1.5ex, -.1ex) -- cycle;%
		}%
	}%
}

\DeclareMathOperator{\Boxdot}{%
	\text{%
		\tikz[baseline]{%
    			\draw[robust] (0ex, -.1ex) -- (0ex, 1.4ex) -- (1.5ex, 1.4ex) -- (1.5ex, -.1ex) -- cycle;%
	    		\fill (.75ex, .65ex) circle (.15ex);%
    		}%
	}%
}

\let\Diamond\relax
\DeclareMathOperator{\Diamond}{%
	\text{%
		\tikz[baseline]{%
			\draw[robust] (0ex,.6ex) -- (.95ex, 1.55ex) -- (1.9ex, .6ex) -- (.95ex, -.35ex) -- cycle;%
		}%
	}%
}

\DeclareMathOperator{\Diamonddot}{%
	\text{%
		\tikz[baseline]{%
			\draw[robust] (0ex,.6ex) -- (.95ex, 1.55ex) -- (1.9ex, .6ex) -- (.95ex, -.35ex) -- cycle;%
			\fill (.95ex, .6ex) circle (.15ex);%
		}%
	}%
}

\DeclareMathOperator{\Next}{%
	\text{%
		\tikz[baseline]{%
    			\draw[robust] (.75ex, .65ex) circle (.75ex);%
    		}%
	}%
}

\DeclareMathOperator{\Nextdot}{%
	\text{%
		\tikz[baseline]{%
    			\draw[robust] (.75ex, .65ex) circle (.75ex);%
	    		\fill (.75ex, .65ex) circle (.15ex);%
    		}%
	}%
}

\newcommand{\stratDiamonddot}[1]{\ensuremath{\langle\!\langle#1\rangle\!\rangle}}
\newcommand{\stratDiamond}[1]{\stratDiamonddot{#1}}

\newcommand{\stratBoxdot}[1]{\ensuremath{[\![#1]\!]}}
\newcommand{\stratBox}[1]{\stratBoxdot{#1}}

\DeclareMathOperator{\cNext}{%
	\text{%
		\tikz[baseline]{%
    			\draw[robust] (.75ex, .65ex) circle (.75ex);%
    		}%
	}%
}

\title{Robust Alternating-Time Temporal Logic}

\author{
    Aniello Murano\inst{1}\orcidID{0000-0003-4876-3448} \and
    Daniel Neider\inst{2,3}\orcidID{0000-0001-9276-6342} \and
    Martin Zimmermann\inst{4}\orcidID{0000-0002-8038-2453}
}

\institute{
    Università degli Studi di Napoli “Federico II”, Naples, Italy\\
    \and
    TU Dortmund University, Germany\\
    \email{daniel.neider@tu-dortmund.de}
    \and
    Center for Trustworthy Data Science and Security, UA Ruhr, Germany \and
    Aalborg University, Denmark
}

\begin{document}

\maketitle

\begin{abstract}
In multi-agent system design, a crucial aspect is to ensure robustness, meaning that for a coalition of agents~$A$, small violations of adversarial assumptions only lead to small violations of $A$'s goals. 
In this paper we introduce a logical framework for robust strategic reasoning about multi-agent systems.
Specifically, inspired by recent works on robust temporal logics, we introduce and study \ratl and \ratlstar, logics that extend the well-known Alternating-time Temporal Logic \atl and \atlstar by means of an opportune multi-valued semantics for the strategy quantifiers and temporal operators. We study the model-checking and satisfiability problems for \ratl and \ratlstar and show that dealing with robustness comes at no additional computational cost. Indeed, we show that these problems are \ptime-complete and \exptime-complete for \ratl, respectively, while both are \twoexptime-complete for \ratlstar.

\end{abstract}

\input{intro}

\input{preliminaries}

\input{ratl}

\input{ratl-model-checking}

\input{ratl-satisfiability}

\input{ratl-expressivity}
\input{ratlstar}

\input{example}

\input{conclusion}

\bibliographystyle{splncs04}
\bibliography{bib}

\clearpage
\appendix
\input{appendix}

\end{document}

%% file: intro.tex

\section{Introduction}
Multi-agent system verification has been receiving a lot of attention in recent years, thanks to the introduction of powerful logics for strategic reasoning~\cite{DBLP:journals/jacm/AlurHK02,DBLP:journals/iandc/ChatterjeeHP10,DBLP:journals/tocl/MogaveroMPV14,DBLP:journals/iandc/LaroussinieM15,DBLP:conf/ijcai/BelardinelliJKM19}. Along this line of research, a story of success  is \emph{Alternating-Time Temporal Logic} (\atl) introduced by Alur, Henzinger, and Kupferman~\cite{DBLP:journals/jacm/AlurHK02}.
\atl is a generalization of Computation Tree Logic (\ctl)~\cite{ctl}, obtained by replacing the path quantifier $\exists$ (and its dual $\forall$), with the modality $\stratDiamonddot{A}$ (and its dual $\stratBoxdot{A}$), where $A$ is a set of agents.
The interpretation of $\stratDiamonddot{A}\varphi$ is that the coalition $A$ has a strategy such that the outcome of this strategy satisfies $\varphi$, no matter how the coalition of the agents not in $A$ behaves. \atl formulas are interpreted over concurrent game structures, which extend classical Kripke structures to represent the dynamism of the agents. The model-checking problem of \atl is \ptime-complete~\cite{DBLP:journals/jacm/AlurHK02}, while the satisfiability problem is \exptime-complete~\cite{walther2006atl}. 

A crucial aspect in multi-agent system design is to ensure system \emph{robustness}, which should reflect the ability of a coalition of agents to tolerate violations (possibly up to some extent) of adversarial assumptions \cite{chaaban2013survey}. Numerous studies have shown that reactive AI systems can be very sensitive to intentional or unintentional external perturbations, posing huge risks to safety-critical applications~\cite{kaur2022trustworthy}.
Notably, the formal methods community has put large efforts in reasoning about system robustness in several specific settings, mainly concerning closed system verification or (two-player) reactive synthesis \cite{bouyer2008robust,french2007temporal,bloem2010robustness,dallal2016synthesis,donze2010robust,doyen2010robustness}.
As far as we are aware of, there are no logic-based works dealing with robust strategic reasoning in multi-agent systems.
To highlight the significance of this challenge, we describe a few specific contexts in which multi-agent systems act as \mbox{the natural model and robustness plays a crucial role.}

\emph{Scenario 1}.
\emph{Climate change} threatens people with food and water scarcity, increased flooding, extreme heat, diseases, and economic loss. Human migration and conflict can be a result. The World Health Organization calls climate change the greatest threat to global health in the 21st century. Recently, researchers examining alternative policies to address the threat of climate change have become increasingly concerned about uncertainty and the fact that we cannot predict the future. 
This requires to develop mathematical models to properly represent the intricate interaction among all decision makers and the ability to define strategies that are robust against a wide range of plausible climate-change futures \cite{lempert2000robust}. 
For risk-averse policy-makers, such strategies would perform reasonably well, at least compared to the alternatives, even if confronted with surprises or catastrophes. Robust strategies may also provide a more solid basis for consensus on political action among stakeholders with different views of the future, because it would provide reasonable outcomes no matter whose view proved correct.


\emph{Scenario 2.} The fast-evolving domain of \emph{autonomous vehicles} is one of the best examples of multi-agent modelling, where  safety-critical decisions strongly rely on sensor observations (e.g., ultrasound, radar, GPS, Lidar, and camera signals)~\cite{veres2011autonomous}. 
It is of primary importance that the resulting decisions are robust to perturbations, which often are treated as adversarial perturbations~\cite{modas2020toward}. A careful evaluation of such adversarial behaviours is necessary to build and deploy safer autonomous vehicle systems. 

\emph{Scenario 3}.
\emph{Power systems} play an important role in all sectors of the national economy and in our daily lives. Ensuring a safe and reliable power supply from the power network is a fundamental requirement. 
As renewable energy-based smart grid and micro-grid systems rise in popularity, multi-agent system technology has been establishing itself as a useful paradigm of choice for modelling, analysis, control and optimization of power systems \cite{hassan2021incorporating,sampaio2017automatic,singh2017distributed}. The model usually consists of several agents competing not only among themselves to get energy resources, but also playing against the unpredictable behaviour of nature.
Then, a classical safety requirement amounts to ensuring system robustness, in the meaning that the power system has to keep operating, possibly by rationing resources, despite the loss of any single asset such as lines or power plants at any time
\cite{afzal2020state,bevrani2014robust}. This is usually enforced by following a simple guiding redundancy principle while developing the system: designers have to predict the effect of having any line disconnected in any moment and cope with it, in real time and even at larger scales \cite{omnes2021adversarial}. This may also require the players to coordinate and/or play rational while keeping the system under equilibrium \cite{belhaiza2014game}. 

\paragraph{Our Contribution}
In this paper we introduce \ratl, a robust version of the logic \atl. 
Our approach here follows and extends an approach originally introduced for robust Linear Temporal Logic (\rltl) \cite{DBLP:conf/csl/TabuadaN16} and later extended to robust Computation Tree Logic (\rctl and \rctlstar) \cite{rctl}. 
To illustrate the robust semantics, consider an invariant of the form $\Box p$ specifying that the proposition~$p$ always holds. 
There are several ways this invariant can be violated, with varying degrees of severity. For example, $p$ failing to hold a finite number of times is less severe than $p$ failing infinitely often. 
An even worse situation is $p$ holding only finitely often while $p$ not even holding once is the worst way to violate the invariant.
The authors in \cite{DBLP:conf/csl/TabuadaN16} argue that these five degrees are canonical and use them as the basis of a five-valued robust semantics for temporal logics.
The semantics of the Boolean operators are then defined to capture the intuition that there are different degrees of violation of a formula while the other temporal operators, e.g., next and eventually, are defined as usual.
In particular, the definition of implications captures the idea that, in a specification of the form $\phi \to \psi$, a ``small'' violation of an environment assumption $\phi$ must lead to only a ``small'' (proportional) violation of a system's guarantee $\psi$.

Here, we devise a meaningful robust semantics for the strategy quantifiers to obtain a robust variant of \atl, and show that it is capable to reason about the robustness of multi-agent systems.
More precisely, \ratl allows to assess whether a strategy $f$ of a coalition $A$ is robust in the sense that, with respect to the outcome of $f$, small violations of the adversarial team assumptions only lead to small violations of $A$’s goals.  
We study expressiveness of \ratl and show that it strictly subsumes \atl, as \ratl can express fairness. We also study the model-checking and satisfiability problems for \ratl and show that dealing with robustness comes at no additional computational cost. Indeed, we show that these problems are \ptime-complete and \exptime-complete, respectively. 
This is in line with the results on \rltl and \rctl, for which model-checking and satisfiability are also not harder than for \ltl~\cite{DBLP:conf/csl/TabuadaN16} and \ctl~\cite{rctl}, respectively.

Finally, we also study \ratlstar, the robustification of \atlstar, showing that also in this setting, robustness comes for free:  model-checking and satisfiability for \ratlstar are $\twoexptime$-complete, as they are already for \atlstar~\cite{DBLP:journals/jacm/AlurHK02,Schewe}.

All proofs omitted due to space restrictions can be found in the appendix.

\paragraph{Related work} 
There are several works done in formal strategic reasoning that have been used (or can be easily adapted) to develop robust systems. Besides those reported above, we recall the works dealing with strategy logics extended with probabilistic \cite{huang2013logic,schnoor2013epistemic,aminof2019probabilistic} and knowledge (imperfect information) aspects \cite{dima2011model}. These works allow to reason about the unpredictable behaviour of the environment. Unfortunately, in both cases, the model-checking problem becomes highly undecidable, unless one restricts strategies to be memoryless. In the imperfect information case, memoryfull strategies with less severe restrictions have been also studied (e.g., hierarchical visibility~\cite{berthon2021strategy} and public action \cite{belardinelli2020verification}) although model-checking remains infeasible, i.e.,  non-elementary, in practice.

Other lines of research have considered quantitative aspects of the logic, in different directions.
Bouyer et al.~\cite{BKMMMP23} considered a fuzzy extension of \atlstar, namely \atlstarf. The satisfaction value of \atlstarf formulas is a real value in~$[0, 1]$, reflecting ``how much'' or ``how well'' the strategic on-going objectives of the underlying
agents are satisfied. In~\cite{BKMMMP23} a double exponential-time model-checking procedure for \atlstarf is presented. A careful inspection of that procedure yields, for the special case of \atlf, an \exptime-completeness result by means of an exponential reduction to Büchi games.
Faella, Napoli, and Parente~\cite{faella2010graded} and Aminof et al.~\cite{aminof2018graded} considered a graded extension of the logics \atl and \atlstar with the ability of checking for the existence of redundant winning strategies.

Module checking is another example of a formal method to devise robust systems. Indeed, module checking amounts to checking whether a strategic behaviour of a coalition of agents satisfies a goal, irrespective to all possible nondeterministic behaviours of an external environment \cite{kupferman2001module,jamroga2014module}. 

Finally, robustness is also an active field of study in reinforcement learning \cite{pinto2017robust}, which treats environment mismatches as adversarial perturbations against a coalition of agents. In the simplest version, the underlying model is a two-player zero-sum
simultaneous game between the protagonist who aims to find a robust strategy across environments and the adversary who
exerts perturbations. Computational methods have been proposed to solve this game and to find a robust strategy for the protagonist
(see Pinto et al.~\cite{pinto2017robust} and the references~therein).


%% file: preliminaries.tex

\section{Preliminaries}
\label{sec:preliminaries}

We denote the nonnegative integers by $\nats$, and the power set of a set~$S$ by $\pow{S}$.
Throughout the paper, we fix a finite set~$\ap$ of atomic propositions. 

A concurrent game structure~$\struct = (\states, \ag, \ac, \tr, \lab)$ consists of a finite set~$\states$ of states, a finite set~$\ag$ of agents, a finite set~$\ac$ of actions, and a labeling function~$\lab \colon \states \rightarrow \pow{\ap}$.
An action vector for a subset~$A \subseteq \ag$ is a mapping~$\vec \colon A \rightarrow \ac$.
Let $\vecs$ denote the set of action vectors for the full set~$\ag$ of agents.
The transition function~$\tr\colon \states \times \vecs \rightarrow \states$ maps a state and an action vector to a state.
The size of $\struct$ is defined as $\size{\states \times \vecs}$.

We say that a state~$s'$ is a successor of a state~$s$ if there is an action vector~$\vec \in \vecs$ such that $s' = \tr(s, \vec)$.
A path of $\struct$ is an infinite sequence~$\pth = s_0s_1s_2\cdots$ of states such that $s_{n+1}$ is a successor of $s_n$ for every $n \ge 0$.
We write $\element{\pth}{n}$ for $s_n$.

A strategy for an agent is a function~$\strat \colon \states^+ \rightarrow \ac$. 
Given a set~$\stratset_A = \set{\strat_a \mid a \in A}$ of strategies, one for each agent in some set~$A \subseteq \ag$, $\outcome(s,\stratset_A)$ denotes the set of paths starting in $s$ that are consistent with $\stratset_A$. 
Formally, a path~$s_0 s_1 s_2 \cdots$ is in $\outcome(s,\stratset_A)$ if $s_0 = s$ and for all $n \ge 0$, there is an action vector~$\vec \in \vecs$ with $\vec(a) = \strat_a(s_0 \cdots s_{n})$ for all $a \in A$ and $s_{n+1} = \delta(s_n, \vec)$.
Intuitively, $\outcome(s, \stratset_A)$ contains all paths that are obtained by the agents in $A$ picking their actions according to their strategies and the other agents picking their actions arbitrarily.

%% file: ratl.tex

\section{\ratl}
\label{sec:rATL}

The basic idea underlying our robust version of ATL, or \emph{\ratl} for short, is that a ``small'' violation of an environment assumption (along the outcome of a strategy) must lead to only a ``small'' violation of a system's guarantee. This is obtained by devising a robust semantics for the strategy quantifiers and by stating formally what it is meant for a ``small'' violations of a property. For the latter, we follow and adapt the approach by Tabuada and Neider~\cite{DBLP:conf/csl/TabuadaN16}, initially proposed for a robust version of Linear Temporal Logic~(\rltl), and use five truth values: $1111$, $0111$, $0011$, $0001$, and $0000$.
Let $\boolfour$ denote the set of these truth values.
Our motivation for using the seemingly odd-looking truth values in $\boolfour$ is that they represent five canonical ways how a system guarantee of the form ``always $p$'' ($\Box p$ in LTL) can be satisfied or violated.
Clearly, we prefer that $p$ always holds, represented by the truth value $1111$.
However, if this is impossible, the following best situation is that $p$ holds at least almost always, represented by $0111$. Similarly, we would prefer $p$ being satisfied at least infinitely often, represented by $0011$, over $p$ being satisfied at least once, represented by $0001$.
Finally, the worst situation is that $p$ never holds, represented by $0000$.
Put slightly differently, the bits of each truth value represent (from left to right) the modalities ``always''~($\Box$), ``eventually always''~($\Diamond\Box$), ``always eventually''~($\Box\Diamond$), and ``eventually''~($\Diamond$).
We refer the reader to Anevlavis et al.~\cite{DBLP:journals/tocl/AnevlavisPNT22} for an in-depth explanation of why these five ways are canonical.

Following the intuition above, we order the truth values in $\boolfour$ by
\[ 1111 \succ 0111 \succ 0011 \succ 0001 \succ 0000. \]
This order spans a spectrum of truth values ranging from $1111$, corresponding to $\mathit{true}$, on one end, to $0000$, corresponding to $\mathit{false}$, on the other end.
Since we arrived at the set $\boolfour$ by considering the canonical ways of how the invariant property $\Box p$ can fail, we interpret all truth values different from $1111$ as \emph{shades of $\mathit{false}$}.
We return to this interpretation when we later define the semantics for the negation in \ratl.

Having formally discussed how we ``grade'' the violation of a property along paths, we are now ready to define the syntax of \ratl via the following grammar:
%
\begin{align*}
	\phi &\Coloneqq p \mid \neg \phi \mid \phi  \vee \phi \mid \phi \wedge \phi \mid \phi \rightarrow \phi \mid \stratDiamonddot{A}\Phi \mid \stratBoxdot{A}\Phi \\
	\Phi  &\Coloneqq \Nextdot\phi \mid \Diamonddot\phi \mid  \Boxdot \phi 
\end{align*}
where $p$ ranges over atomic propositions and $A$ ranges over subsets of agents.
We distinguish between \emph{state formulas} (those derivable from $\varphi$) and \emph{path formulas} (those derivable from $\Phi$).
If not specified, an \ratl formula is a state formula.

Various critical remarks should be made concerning the syntax of \ratl.
First, we add ``dots'' to temporal operators (following the notation by Tabuada and Neider~\cite{DBLP:conf/csl/TabuadaN16}) to distinguish between the original operators in \atl and their robustified counterparts in \ratl---otherwise, the syntax stays the same.
Second, many operators of \ratl, most notably the negation and implication, can no longer be derived via De\,Morgan's law or simple logical equivalencies due to \ratl's many-valued nature.
Hence, they need to be added explicitly.
Third, we omit the until and release operators here to avoid cluttering our presentation too much.
Both can be added easily, as in \rltl~\cite{DBLP:conf/cdc/AnevlavisPNT18,DBLP:conf/hybrid/AnevlavisNPT19}.

We define the semantics of \ratl by an \emph{evaluation function}~$\eval$ that maps a state formula and a state or a path formula and a path to a truth value in $\boolfour$.
To simplify our presentation, we use $b[k]$ as a shorthand notation for addressing the $k$-th bit, $k \in \set{1,2,3,4}$, of a truth value $b = b_1b_2b_3b_4 \in \boolfour$ (i.e., $b[k] = b_k$).
It is worth emphasizing that our semantics for \ratl is a natural extension of the Boolean semantics of \atl and is deliberately designed to generalize the original Boolean semantics of \atl (see Subsection~\ref{subsec_expressiveness}).

Turning to the definition of \ratl's semantics, let us begin with state formulas.
For atomic propositions $p \in \ap$, we define the valuation function by
\[
	\eval(s, p) = \begin{cases} 1111 &\text{if $p \in \lab(s)$; and} \\ 0000 & \text{if $p \notin \lab(s)$.} \end{cases}
\]
Note that this definition mimics the semantics of \atl in that propositions get mapped to one of the two truth values $\mathit{true}$ ($1111$) or $\mathit{false}$ ($0000$).
As a consequence, the notion of robustness in \ratl does not arise from atomic propositions (e.g., as in \ltlf by Almagor, Boker, and Kupferman~\cite{DBLP:journals/jacm/AlmagorBK16} or fuzzy logics) but from the evolution of the temporal operators (see the semantics of path formulas).
This design choice is motivated by the observation that assigning meaningful (robustness) values to atomic propositions is often highly challenging in practice---if not impossible.

The semantics of conjunctions and disjunctions are defined as usual for many-valued logics in terms of the functions $\min$ and $\max$:
\begin{align*}
	\eval(s, \varphi_1 \lor \varphi_2) & = \max{\bigl( \eval(s, \varphi_1), \eval(s, \varphi_2) \bigr)} \\    
	\eval(s, \varphi_1 \land \varphi_2) & = \min{ \bigl( \eval(s, \varphi_1), \eval(s, \varphi_2) \bigr)}
\end{align*}

To define the semantics of negation, remember our interpretation of the truth values in $\boolfour$: $1111$ corresponds to $\mathit{true}$ and all other truth values correspond to different shades of $\mathit{false}$.
Consequently, we map $1111$ to $0000$ and all other truth values to $1111$. 
This idea is formalized by
\[
	\eval(s, \lnot \phi) = \begin{cases} 0000 & \text{if $\eval(s,\phi) = 1111$; and} \\ 1111 & \text{if $\eval(s, \phi) \prec 1111$.} \end{cases}
\]
Note that the definition of $\eval(s, \lnot \phi)$ is not symmetric, which is in contrast to other many-valued logics, such as \ltlf.
However, it degenerates to the standard Boolean negation if one considers only two truth values.

Since our negation is defined in a non-standard way, we cannot recover implication from negation and disjunction.
Instead, we 
define the implication $a \rightarrow b$ by requiring that $c \prec a \rightarrow b$ if and only if $\min{\{ a, c \}} \prec b$ for every $c \in \boolfour$.
This notion leads to
\[
	\eval(s, \varphi_1 \rightarrow \varphi_2) = \begin{cases} 1111 & \text{if $\eval(s, \varphi_1) \preceq \eval(s, \varphi_2)$; and} \\ \eval(s,\varphi_2) & \text{if $\eval(s, \varphi_1) \succ \eval(s, \varphi_2)$.} \end{cases}
\]
Again, this definition collapses to the usual Boolean definition in case one considers only two truth values.

We now provide the robust semantics for the strategy quantifiers, which
are the key ingredient in \ratl. First, notice that the strategy quantifiers~$\stratDiamonddot{\cdot}\,$ and $\stratBoxdot{\cdot}$ are not dual in our robustified version of \atl and require their individual definitions.
Intuitively, $\stratDiamonddot{A}\Phi$ is the largest truth value that the coalition~$A$ of agents can enforce for the path formula~$\Phi$, while $\stratBoxdot{A}\Phi$ is the largest truth value that $\ag\setminus A$ can enforce against $A$. Formally, we have the following:
\begin{itemize}
	\item $\eval(s, \stratDiamonddot{A}\Phi)$ is the maximal truth value~$b \in \boolfour$ such that there is a set~$\stratset_A$ of strategies, one for each agent in $A$, such that for all paths~$\pth \in \outcome(s, \stratset_A)$ we have $\eval(\pth, \Phi) \succeq b$.
	\item $\eval(s, \stratBoxdot{A}\Phi)$ is the maximal truth value~$b \in \boolfour$ such that for all sets~$\stratset_A$ of strategies, one for each agent in $A$, there exists a path~$\pth \in \outcome(s, \stratset_A)$ with $\eval(\pth, \Phi) \succeq b$.
\end{itemize}

Let us now turn to the 
semantics of path formulas.
We begin with the $\Boxdot$-operator.
This operator captures the five canonical ways an invariant property ``always~$p$'' can be satisfied or violated, thereby implementing the
intuition we have presented at the beginning of this section.
Formally, the valuation function $\eval(\pth, \Boxdot \phi)$ is given by $\eval(\pth, \Boxdot \phi) = b_1b_2b_3b_4$ where
\begin{align*}
    b_1 &= \min\nolimits_{i \ge 0} \eval(\pth[i], \phi)[1], & b_3 &= \min\nolimits_{i \ge 0} \max\nolimits_{j \ge i} \eval(\pth[j], \phi)[3],\\
    b_2 &= \max\nolimits_{i \ge 0} \min\nolimits_{j \ge i} \eval(\pth[j], \phi)[2],& b_4 &= \max\nolimits_{i \ge 0} \eval(\pth[i], \phi)[4] ).
\end{align*}
Note that for $p \in \ap$ and a path $\pi$, the semantics of the formula $\Boxdot p$ on $\pi$ amounts to the four-tuple $(\Box p$, $\Diamond\Box p$, $\Box\Diamond p$, $\Diamond p)$ because $\eval(s, p)$ is either $0000$ or $1111$ on every state $s$ along $\pi$ (i.e., all bits are either $0$ or $1$).
However, the interpretation of $\eval(\pi, \Boxdot \phi)$ becomes more involved once the formula~$\phi$ is nested since the semantics of the $\Boxdot$-operator refers to individual bits of $\eval(\pi, \phi)$.

Finally, the semantics for the $\Diamonddot$-operator and $\Nextdot$-operator are straightforward as there are only two possible outcomes: either the property is satisfied, or it is violated.
Consequently, we define the valuation function by
\begin{itemize}
	\item \mbox{$\eval(\pth, \Diamonddot \phi) = b_1b_2b_3b_4$ with $b_k = \max_{i \ge 0} \eval(\pth[i], \phi)[k]$;} and
    \item $\eval(\pth, \Nextdot \phi) = b_1 b_2 b_3 b_4$ with $ b_k = \eval(\pth[1], \phi)[k]$.
\end{itemize}
Again, note that both $\eval(\pth, \Diamonddot \phi)$ and $\eval(\pth, \Nextdot \phi)$ refer to individual bits of $\eval(\pi, \phi)$.

\begin{example}
\label{example:ratl}
Consider the formula~$\phi = \stratDiamonddot{A}\Boxdot p$. 
We have 
\begin{itemize}
    \item $\eval(s, \phi) = 1111$ if the coalition~$A$ has a (joint) strategy to ensure that $p$ holds at every position of every outcome. 
    \item $\eval(s, \phi) = 0111$ if the coalition~$A$ has strategy to ensure that $p$ holds at all but finitely many positions of every outcome. 
    \item $\eval(s, \phi) = 0011$ if the coalition~$A$ has strategy to ensure that $p$ holds at infinitely many positions of every outcome. 
    \item $\eval(s, \phi) = 0001$ if the coalition~$A$ has strategy to ensure that $p$ holds at least once on every outcome. 
\end{itemize}
\end{example}

%% file: ratl-model-checking.tex

\subsection{\ratl Model-Checking}
\label{subsec:modelchecking}
The model-checking problem for \ratl is as follows: Given a concurrent game structure~$\struct$, a state~$s$, an \ratl formula~$\varphi$, and a truth value~$\tval \in \boolfour$, is $\eval(s, \varphi) \succeq \tval$?

\begin{theorem}
\label{thm:mc}
\ratl model-checking is \ptime-complete.
\end{theorem}

The proof is based on capturing the semantics of the strategy quantifiers~$\stratDiamonddot{A}$ and $\stratBoxdot{A}$ by sequential two-player games, one player representing the agents in $A$ and the other representing the agents in the complement of $A$. 
We begin by introducing the necessary background on such games.

A (sequential) two-player game structure~$\struct = (\states, \states_1, \states_2, \ac_1, \ac_2, \tr)$ consists of a set~$\states$ of states partitioned into the states~$\states_p \subseteq \states$ of Player~$p \in \set{1,2}$, an action set~$\ac_p$ for Player~$p \in \set{1,2}$, and a transition function~$\tr \colon \states_1 \times \ac_1 \cup \states_2\times \ac_2 \rightarrow \states$.
The size of $\struct$ is $\size{ \states_1 \times \ac_1 \cup \states_2\times \ac_2}$.
A path of $\struct$ is an infinite sequence~$s_0 s_1 s_2 \cdots$ of states such that $s_{n+1} = \tr(s_n, \alpha)$ for some action~$\alpha$.
A strategy for Player~$1$ is a mapping~$\strat \colon \states^\ast\states_1 \rightarrow \ac_1$.
A path~$s_0s_1s_2 \cdots$ is an outcome of $\strat$ starting in $s$, if $s_0 = s$ and  $s_{n+1}= \tr(s_n, \strat(s_0 \cdots s_n))$ for all $n\ge 0$ such that $s_n \in \states_1$.
A two player game~$\game = (\struct, \win)$ consists of a two-player game structure~$\struct$ and a winning condition~$\win \subseteq \states^\omega$, where $\states$ is the set of states of $\struct$.
We say that a strategy~$\strat$ for Player~$1$ is a winning strategy for $\game$ from a state~$s$, if every outcome of $\strat$ starting in $s$ is in $\win$.

Given a concurrent game structure~$\struct = (\states, \ag, \ac, \tr, \lab)$ and $A \subseteq \ag$, we define the two-player game structure~$\struct_A = (\states_1 \cup \states_2, \states_1, \states_2, \ac_1, \ac_2, \tr')$ where $\states_1 = \states$ and $\states_2 = \states \times \ac_1$,
 $\ac_1$ is the set of action vectors for $A$,
 $\ac_2$ is the set of action vectors for $\ag \setminus A$,
 $\tr'(s, \vec) = (s, \vec)$ for $s \in \states_1$ and $\vec \in \ac_1$, and
  $\tr'((s, \vec), \vec') = \tr(s, \vec\oplus\vec')$ for $(s,\vec) \in \states_2$ and $\vec' \in \ac_2$, where $\vec\oplus\vec'$ is the unique action vector for $\ag$ induced by $\vec$ and $\vec'$.
Note that the size of $\struct_A$ is at most linear in the size of $\struct$.

A path in $\struct_A$ alternates between states of $\struct$ and auxiliary states (those in $\states \times \ac_1$), i.e., it is in $(\states \cdot(\states \times\ac_1))^\omega$.
Thus, when translating paths between $\struct$ and $\struct_A$, only states at even positions are relevant (assuming we start the path in $\struct_A$ in $\states$).
Hence, given a property~$P \subseteq \states^\omega$ of paths in $\struct$, we extend it to the corresponding winning condition~$P' = \set{s_0 s_1 s_2 \cdots \in (\states \cdot( \states \times\ac_1))^\omega \mid s_0 s_2 s_4\cdots \in P }$ of paths in $\struct_A$.

The next lemma reduces the (non-) existence of strategies that allow a set~$A$ of agents to enforce a property in $\struct$ (which formalize the semantics of $\stratDiamonddot{A}$ and $\stratBoxdot{A}$) to the (non-) existence of winning strategies for Player~$1$ in $\struct_A$.
It derives 
from results of de~Alfaro and Henzinger~\cite{concomegagames} for concurrent $\omega$-regular games.

\begin{lemma}
\label{lemma:gamereduction}
Let $\struct$ be a concurrent game structure with set~$\states$ of states containing $s$, let $A$ be a subset of its agents, and let $P \subseteq \states^\omega$.
\begin{enumerate}
    \item There is a set~$\stratset_A$ of strategies, one for each agent~$a \in A$, such that $\outcome(s, \stratset_A) \subseteq P$ iff Player~$1$ has a winning strategy for $(\struct_A, P')$ from $s$.
    \item For all sets~$\stratset_A$ of strategies, one for each agent~$a \in A$, $\outcome(s, \stratset_A) \cap P \neq \emptyset$ iff Player~$1$ does not have a winning strategy for $(\struct_A, (\states^\omega\setminus P)')$ from $s$.
\end{enumerate}
\end{lemma}

In the following, we consider the following winning conditions for a two-player game played in $\struct_A$, all induced by a set~$F \subseteq \states $ of states:
\begin{align*}
    \next(F) & = \set{s_0s_1s_2\cdots \in (\states \cdot( \states \times\ac_1))^\omega \mid s_2 \in F} \\
    \reach(F) & = \set{s_0s_1s_2\cdots \in (\states \cdot (\states \times\ac_1))^\omega \mid s_n \in F \text{ for some even }n} \\
    \safe(F) & = \set{s_0s_1s_2\cdots \in (\states \cdot (\states \times\ac_1))^\omega \mid s_n \in F \text{ for all even }n} \\
    \buchi(F) & = 
    \begin{multlined}[t]
        \set{s_0s_1s_2\cdots \in (\states \cdot (\states \times\ac_1))^\omega \mid {} \\
        s_n \in F \text{ for infinitely many even }n}
    \end{multlined} \\
    \cobuchi(F) & =
    \begin{multlined}[t]
        \set{s_0s_1s_2\cdots \in (\states \cdot (\states \times\ac_1))^\omega \mid {} \\
        s_n \in F \text{ for all but finitely many even }n}
    \end{multlined}
\end{align*}

Again, note that these conditions only refer to even positions, as they will be used to capture a property of paths in $\struct$, i.e., the auxiliary states are irrelevant.


Collectively, we refer to games with any of the above winning conditions as \nrsbc games.
The following result is a generalization of standard results on infinite games (see, e.g., Grädel, Thomas, and Wilke~\cite{gtw}) that accounts for the fact that only states at even positions are relevant.

\begin{proposition}
\label{prop:nrsbcgamecomplexity}
The following problem is in \ptime: Given an \nrsbc game~$\game$ and a state~$s$, does Player~$1$ have a winning strategy for~$\game$ from~$s$?
\end{proposition}

\begin{proof}[Proof of Theorem~\ref{thm:mc}]
Consider a concurrent game structure~$\struct$ with set~$\states$ of states and an \ratl formula~$\varphi$.
We show how to inductively compute the satisfaction sets~$\satset(\varphi',\tval) = \set{s \in \states \mid \eval(s, \varphi') \succeq \tval}$ for all (state) subformulas~$\varphi'$ of $\varphi$ and all truth values~$\tval \in\boolfour$.
Note that $\satset(\varphi',0000) = \states$ for all formulas~$\phi'$, so these  sets can be computed trivially.

The cases of atomic propositions and Boolean connectives follow straightforwardly from the definition of their semantics (cp.\ the semantics of rCTL~\cite{rctl}), so we focus on the case of formulas of the form~$\stratDiamonddot{A}\Phi$ or $\stratBoxdot{A}\Phi$. 
Note that we only have to consider three cases for $\Phi$, e.g., $\Phi = \Nextdot\phi'$, $\Phi = \Diamonddot\phi'$, and $\Phi = \Boxdot\phi'$ for some state formula~$\phi'$.
The following characterizations are consequences of Lemma~\ref{lemma:gamereduction}:

\begin{itemize}
    \item $s \in \satset(\stratDiamonddot{A}\Nextdot\phi', \tval)$ if and only if Player~$1$ has a winning strategy for $(\struct_A,\next(\satset(\phi', \tval)))$ from $s$.
    
    \item $s \in \satset(\stratDiamonddot{A}\Diamonddot\phi', \tval)$ if and only if Player~$1$ has a winning strategy for $(\struct_A,\reach(\satset(\phi', \tval)))$ from $s$.
    
    \item $s \in \satset(\stratDiamonddot{A}\Boxdot\phi', 1111)$ if and only if Player~$1$ has a winning strategy for $(\struct_A,\safe(\satset(\phi', 1111)))$ from $s$.
    
    \item $s \in \satset(\stratDiamonddot{A}\Boxdot\phi', 0111)$ if and only if Player~$1$ has a winning strategy for $(\struct_A,\cobuchi(\satset(\phi', 0111)))$ from $s$.
    
    \item $s \in \satset(\stratDiamonddot{A}\Boxdot\phi', 0011)$ if and only if Player~$1$ has a winning strategy for $(\struct_A,\buchi(\satset(\phi', 0011)))$ from $s$.
    
    \item $s \in \satset(\stratDiamonddot{A}\Boxdot\phi', 0001)$ if and only if Player~$1$ has a winning strategy for $(\struct_A,\reach(\satset(\phi', 0001)))$ from $s$.


    

    
    
\end{itemize}
Analogously, the satisfaction of formulas~$\stratBoxdot{A}\Phi$ can be characterized by the non-existence of winning strategies for Player~$1$, relying on the duality of the reachability (Büchi) and safety (coBüchi) winning conditions and the self-duality of the winning condition capturing the next operator. 
For example, we have $s \in \satset(\stratBoxdot{A}\Nextdot\phi',\tval)$ if and only if Player~$1$ does not have a winning strategy for $(\struct_A, \next(\states \setminus \satset(\phi',\tval)))$ from $s$.

Now, to solve the model-checking problem with inputs~$\struct$, $\varphi$, $s$ and $t$, we inductively compute all satisfaction sets~$\satset(\varphi',t')$ and check whether $s$ is in $\satset(\varphi, t)$.
Using Proposition~\ref{prop:nrsbcgamecomplexity} and the fact that each \nrsbc game we have to solve during the computation is of linear size (in $\size{\struct}$), these $\mathcal{O}(\size{\varphi} \cdot \size{\struct})$ many sets can be computed in polynomial time, where $\size{\phi}$ is the number of state subformulas of $\phi$.

Finally, the lower bound follows from the \ptime-hardness of \ctl model-checking~\cite{ctlmc}, which is a fragment of \ratl (see Subsection~\ref{subsec_expressiveness}).
Furthermore, let us note that the \ptime lower bound for \ctl model-checking already holds for fragment without until and release~\cite{krebsetal} (recall that we do not include until and release in \ratl for the sake of simplicity). 
\end{proof}

%% file: ratl-satisfiability.tex
\subsection{\ratl Satisfability}
\label{subsec:satisfiability}

This subsection considers the satisfiability problem for \ratl, which is stated as follows: Given
an \ratl formula~$\varphi$ and a truth value~$\tval \in\boolfour$, is there a concurrent game structure~$\struct$ with a state~$s$ such that $\eval(s, \varphi) \succeq t$?


\begin{theorem}
\label{thm:ratlsat}
\ratl satisfiability is \exptime-complete.
\end{theorem}

\begin{proof}[Proof sketch]
The upper bound is proven by embedding \ratl into the alternating $\mu$-calculus while the lower bound already holds for \ctl, a fragment of \ratl.
\end{proof}

%% file: ratl-expressivity.tex

\subsection{Expressiveness}
\label{subsec_expressiveness}

The main impetus for introducing \ratl is to devise a robust generalization of \atl as a powerful formalism to deal with robust strategic reasoning in multi-agent systems. A natural question is to state the expressive power of \ratl  with respect to \atl and the robust version of \ctl (\rctl)~\cite{rctl}. 
In this subsection, we show that both \atl and \rctl can be embedded into \ratl, i.e., \ratl generalizes both of these logics.
Furthermore, we show that \ratl is strictly more expressive than both of them.
We begin by comparing \ratl and \atl, and show first that \ratl is at least as expressive as \atl, witnessing that our robust extension is set up correctly.
This fact is formalized in the lemma below, intuitively stating that the first bit of the evaluation function captures the semantics of \atl.

\begin{lemma} \label{lemma:ratl-as-expressive-as-atl}
Let $\phi$ be an \atl formula. Then, there exists an \ratl formula $\phi^\star$ such that for every concurrent game structure~$\struct$ and all states~$s$ of $\struct$: $\eval(s, \phi^\star) = 1111$ if and only if $\struct, s \models \phi$.
\end{lemma}

\begin{proof}[Proof sketch]

We obtain the \ratl formula $\phi^\star$ as follows:
First, we eliminate every implication $\phi_1 \rightarrow \phi_2$ in the \atl formula $\phi$ by replacing it with the expression~$\lnot \phi_1 \lor \phi_2$.
Second, we bring the formula into negation normal form by pushing all negations inwards to the level of atomic propositions.
Finally, we dot all the temporal operators to obtain the \ratl formula~$\phi^\star$.
The claim of Lemma~\ref{lemma:ratl-as-expressive-as-atl} can then be shown by induction over the structure of $\phi$.
\end{proof}

As we have observed above with Example~\ref{example:ratl}, \ratl is able to express basic forms of fairness such as "for a given structure~$\struct$ there exists a strategy for a coalition of agents $A$ such that a certain property $p$ holds infinitely often". Formally this corresponds to the formula $\varphi=\stratDiamonddot{A}\Boxdot p$ with $\eval(s, \varphi) \succeq 0011$.
As shown by Alur, Henzinger, and Kupferman~\cite{DBLP:journals/jacm/AlurHK02}, such a property cannot be expressed in \atl, but rather requires the more expressive logic \atlstar. Indeed, it corresponds to the \atlstar formula $\varphi=\stratDiamonddot{A}\Box\Diamond p$. So, by using the result reported in Lemma~\ref{lemma:ratl-as-expressive-as-atl}, the following holds. 

\begin{theorem}\label{thm:ratl-strictly-more-expressive-atl}
\ratl is strictly more expressive than \atl.
\end{theorem}

Now, we compare \ratl and \rctl:
The latter logic is obtained by robustifying \ctl along the same lines as described in Section~\ref{sec:rATL} (see~\cite{rctl} for detailed definitions).
Let us just remark that \rctl formulas, as \ctl formulas, are evaluated over Kripke structures by means of a valuation function $\eval_{\rctl}$. 
Thus, to compare the expressiveness of both logics, as usual, we have to interpret a Kripke structure as a (one-agent) concurrent game structure. 
We start by showing that \ratl is at least as expressive as \rctl, just as \atl is at least as expressive as \ctl.

\begin{lemma} \label{lemma:ratl-as-expressive-as-rctl}
Let $\phi$ be an \rctl formula. Then, there exists an \ratl formula $\phi^\star$ such that for every Kripke structure~$\trans$ the following holds for all states~$s$ of $\trans$: $\eval(s, \phi^\star) = \eval_{\rctl}(s, \phi)$.
\end{lemma}

\begin{proof}[Proof sketch]
Our construction proceeds as follows:
First, we turn a Kripke structure~$\trans$ into a concurrent game structure with one agent~$a$, having the same states and state labels, a suitable set of actions, and a transition function~$\delta$ such that there is a transition in $\trans$ from $s$ to $s'$ if and only if $s' = \delta(s, \alpha)$ for some action~$\alpha$.
Second, we replace each existential path quantifier~$\exists$ in $\phi$ by $\stratDiamonddot{\set{a}}$ and each universal path quantifier~$\forall$ by $\stratDiamonddot{\emptyset}$, obtaining the \ratl formula~$\phi^\star$.
The claim of Lemma~\ref{lemma:ratl-as-expressive-as-rctl} can then be shown by induction over the structure of $\phi$.
\end{proof}

Now, we recall that Alur, Henzinger, and Kupferman~\cite{DBLP:journals/jacm/AlurHK02} have observed that in \atl there are formulas that cannot be expressed in \ctl. The reason is that, given a concurrent game structure, \ctl  can only reason about a single path (with the existential modality) or all paths (with the universal modality). Conversely, \atl can reason about an arbitrary number of paths by means of strategies. The same argument can be extend to \ratl and \rctl. Thus, by putting together this observation with the statement of Lemma~\ref{lemma:ratl-as-expressive-as-rctl}, the following holds. 

\begin{theorem}\label{thm:ratl-strictly-more-expressive-rctl}
\ratl is strictly more expressive than \rctl.
\end{theorem}

Notice that 
the argument that \ratl formulas expressing fairness properties such as "infinitely often" cannot be expressed in \atl (used in Theorem \ref{thm:ratl-strictly-more-expressive-atl} for the strict containment of \atl in \ratl) can also be applied to \rctl. 
Similarly, the argument used above to show that \ratl formulas cannot be translated into \rctl (used in Theorem \ref{thm:ratl-strictly-more-expressive-rctl} for the strict containment of \rctl in \ratl) can also be applied to \atl. This leads to the following corollary.

\begin{corollary}\label{cor:atl-and-rctl-incomparable}
\atl and \rctl are incomparable.
\end{corollary}

%% file: ratlstar.tex
\section{Robust ATL*}
Just as one generalizes \ctl, \rctl, and \atl by allowing nesting of temporal operators in the scope of a single path/strategy quantifier (obtaining \ctlstar, \rctlstar, and \atlstar, respectively), we now study \ratlstar, the analogous generalization of \ratl.
Again, we will prove that adding robustness comes for free.

The formulas of \ratlstar are given by the grammar
\begin{align*}
	\phi & \Coloneqq p \mid \neg \phi \mid \phi  \vee \phi \mid \phi \wedge \phi \mid \phi \rightarrow \phi \mid \stratDiamonddot{A}\Phi \mid \stratBoxdot{A}\Phi \\
	\Phi & \Coloneqq \phi \mid \neg \Phi \mid \Phi  \vee \Phi \mid \Phi \wedge \Phi \mid \Phi \rightarrow \Phi \mid \Nextdot\Phi \mid \Diamonddot\Phi \mid  \Boxdot \Phi 
\end{align*}
where $p$ ranges over atomic propositions and $A$ over subsets of agents.
Again, we distinguish between \emph{state formulas} (those derivable from $\varphi$) and \emph{path formulas} (those derivable from $\Phi$).
If not specified, an \ratlstar formula is a state formula.

The semantics of \ratlstar are again defined via an evaluation function~$\eval$ that maps a state formula and a state or a path formula and a path to a truth value in $\boolfour$. 
The cases for state formulas are defined as for \ratl and we define for every path~$\pth$, every state formula~$\phi$, and all path formulas~$\Phi_1$ and $\Phi_2$
\begin{itemize}
    \item $\eval(\pth,\phi) = \eval(\pth[0],\phi)$, 
    \item $\eval(\pth,\neg \Phi ) = \begin{cases} 0000 & \text{if $\eval(\pth,\Phi) = 1111$,} \\ 1111 & \text{if $\eval(\pth, \Phi) \prec 1111$,} \end{cases}$
    \item $\eval(\pth,\Phi_1  \vee \Phi_2 )=\max{\bigl( \eval(\pth, \Phi_1), \eval(\pth, \Phi_2) \bigr)}$,    
    \item $\eval(\pth,\Phi_1 \wedge \Phi_2)=\min{\bigl( \eval(\pth, \Phi_1), \eval(\pth, \Phi_2) \bigr)}$,
    \item $\eval(\pth,\Phi_1 \rightarrow \Phi_2 )= \begin{cases} 1111 & \text{if $\eval(\pth, \Phi_1) \preceq \eval(\pth, \Phi_2)$,} \\ \eval(\pth,\Phi_2) & \text{if $\eval(\pth, \Phi_1) \succ \eval(s, \Phi_2)$,} \end{cases}$
    \item $\eval(\pth,\Nextdot\Phi )= b_1b_2b_3b_4$ with $b_k = \eval(\pth[1],\Phi)[k]$,
    \item $\eval(\pth,\Diamonddot\Phi)=b_1b_2b_3b_4$ with $b_k = \max_{i\ge 0} \eval(\pth[i],\Phi)[k]$, and  
    \item $\eval(\pth,\Boxdot \Phi)=b_1b_2b_3b_4$ where
\begin{align*}
\textstyle	b_1 &=\textstyle \min_{i \ge 0} \eval(\pth[i], \Phi)[1], & 
b_3 &= \textstyle \min_{i \ge 0} \max_{j \ge i} \eval(\pth[j], \Phi)[3],\\
	 b_2 & \textstyle = \max_{i \ge 0} \min_{j \ge i} \eval(\pth[j], \Phi)[2], & 
b_4 &= \textstyle\max_{i \ge 0} \eval(\pth[i], \Phi)[4] ).
\end{align*}
\end{itemize}

We show that every \ratlstar formula (w.r.t.\ a fixed truth value) can be translated into an equivalent \atlstar formula of polynomial size. 
This allows us to settle the complexity of \ratlstar model-checking and satisfiability as well as the expressiveness of \ratlstar.
Below, $\models$ denotes the \atlstar satisfaction relation~\cite{DBLP:journals/jacm/AlurHK02}.

\begin{lemma}
\label{lemma:ratlstar2atlstar}
For every \ratlstar formula~$\phi$ and every truth value~$t \in \boolfour$, there is an \atlstar formula~$\phi_t$ such that $\eval(s, \phi) \succeq t$ if and only if $\struct, s \models \phi_t$.
Furthermore, the function mapping $\phi$ and $t$ to $\phi_t$ is polynomial-time computable.
\end{lemma}

The \ratlstar model-checking and satisfiability problems are defined as their counterparts for \ratl. 
Both model-checking and satisfiability for \atlstar are $\twoexptime$-complete~\cite{DBLP:journals/jacm/AlurHK02,Schewe}. 
Due to Lemma~\ref{lemma:ratlstar2atlstar}, we obtain the same results for \ratlstar, thereby showing that adding robustness comes indeed for free.

\begin{theorem}\hfill
\label{thm:ratlstarresults}
The \ratlstar model-checking problem and the \ratlstar satisfiability problem are both \twoexptime-complete.
\end{theorem}

Another consequence of the translation from \ratlstar to \atlstar and the fact that \atlstar is a fragment of \ratlstar is that both logics are equally expressive.

\begin{corollary}
\ratlstar and \atlstar are equally expressive.
\end{corollary}

%% file: example.tex

\section{A Practical Example}
\label{sec:example}
%
%
Let us consider a smart grid with a set~$U$ of utility companies and a set $C$ of consumers. Assume that for every consumer $c \in C$ there is a proposition ${\ell}_c$ indicating that $c$'s energy consumption is within the pre-agreed limit.
Conversely, $c$'s consumption is higher than the limit if $\ell_c$ is violated.
Furthermore, there is a proposition~``$\mathrm{stable}$'' that holds true if and only if the grid is stable (i.e., the utility companies coordinate to provide the right amount of electricity).

Let us now consider the \atlstar formula
\[
\stratDiamonddot{U}\stratBoxdot{C} (\Box \bigwedge\nolimits_{c \in C} \ell_c) \rightarrow \Box \mathrm{stable}.
\]
This formula expresses that the utility companies~$U$ have a strategy such that no matter how the consumers behave, the following is satisfied: if each consumer's consumption always stays within their limit, then the utility companies keep the grid always stable. 
However, this specification is not robust and provides only limited information when satisfied: even if a single consumer exceeds their limit once, there is no further obligation on the utility companies, and the formula is satisfied independently of whether the grid is always stable or not.

So, let us illustrate how the \ratlstar formula 
\[
\varphi = \stratDiamonddot{U}\stratBoxdot{C} (\Boxdot \bigwedge\nolimits_{c \in C} \ell_c) \rightarrow \Boxdot \mathrm{stable}
\]
does capture robustness.
To this end, assume for now that $\varphi$ evaluates to $1111$. Then, there is a strategy for $U$ such that for all outcomes~$\pi$ that are consistent with that strategy, the following holds: 
\begin{itemize}
    
    \item If $\bigwedge_{c \in C} \ell_c$ holds in every position of $\pi$, i.e.,
    $\Boxdot \bigwedge_{c \in C} \ell_c$ evaluates to $1111$ then by the semantics of $\rightarrow$ the formula $\Boxdot \mathrm{stable}$ also evaluates to $1111$. This means the proposition ``$\mathrm{stable}$'' also holds in every position. Therefore, the grid supply is always stable. Hence, the desired goal 
    is retained when the assumption regarding the consumers holds with no violation. Note that this is equivalent to what the original \atlstar formula above expresses.
    
    \item Assume now that the consumer assumption~$\bigwedge_{c \in C} \ell_c$ is violated finitely many times, i.e., finitely often some consumer violates their consumption limit. This means that the formula $\Boxdot \bigwedge_{c \in C} \ell_c$ evaluates to $0111$. Then, by the semantics of \ratlstar, $\Boxdot \mathrm{stable}$ evaluates to $0111$ or higher, which means that ``$\mathrm{stable}$'' holds at every state, except for a finite number of times. So, the degree of violation of the guarantee required by $U$ is at most the degree of violation of the assumptions on the consumers.
    
    \item Similarly, if $\Boxdot \bigwedge_{c \in C} \ell_c$ holds  infinitely (finitely) often, then $\Boxdot \mathrm{stable}$ holds infinitely (finitely) often. 

\end{itemize}
If the formula~$\varphi$ evaluates to $1111$, then $U$ has a strategy that does not behave arbitrarily in case the assumption~$\Boxdot \bigwedge_{c \in C} \ell_c$ fails, but instead satisfies the guarantee~$\Boxdot \mathrm{stable}$ to at least the same degree that the guarantee holds.

Finally, even if $\varphi$ evaluates to a truth value $t \prec 1111$, this reveals crucial information about $U$'s ability to guarantee a stable grid, i.e., the premise $\Boxdot \bigwedge_{c \in C} \ell_c$ evaluates to some truth value~$t' \succ t$ while the conclusion~``$\mathrm{stable}$'' evaluates to $t$.

%% file: conclusion.tex
\section{Discussion and Future Work}
\label{sec:conc}

This paper introduces \ratl and \ratlstar, the first logic formalisms able to deal with robust strategic reasoning in multi-agent systems. As we have shown along the paper, \ratl results to be very expressive, useful in practice, and not more costly than the subsumed logics \atl and \rctl. Similarly, \ratlstar is not more costly than the subsumed logic \atlstar.

The positive results about \ratl represent the foundation for a number of useful extensions, mainly by extending robustness to logics for strategic reasoning that are more expressive than \atl and \atlstar such as Strategy Logic \cite{DBLP:journals/tocl/MogaveroMPV14} and the like. 
Notably, Strategy Logic is much more expressive than \atlstar \cite{DBLP:journals/jacm/AlurHK02}. Indeed it can express several game-theoretic concepts including Nash Equilibria over \ltl goals. Interestingly, the formula expressing Nash Equilibria uses an implication. In words the formula says that $n$ agents' strategies $\sigma_1, \ldots, \sigma_n$ form an equilibrium if, for every agent, it holds that whenever by unilaterally changing her strategy the goal is also satisfied, then it \emph{implies} that the goal is satisfied with the original tuple of strategies as well. Robustness in Strategy Logic (by means of \rltl goals in place of \ltl) then allows to define a stronger notion of Nash Equilibrium.

Another interesting direction for future work is to come up with an implementation of the model-checking procedure for \ratl, possibly by extending existing tools such as MCMAS~\cite{MCMAS,CLM15}.

\paragraph*{Acknowledgments.}
This research has been supported by the PRIN project RIPER (No. 20203FFYLK), the PNRR MUR project PE0000013-FAIR, the InDAM project ``Strategic Reasoning in Mechanism Design'', and DIREC - Digital Research Centre Denmark.
Furthermore, this work has been financially supported by Deutsche Forschungsgemeinschaft, DFG Project numbers 434592664 and 459419731, and the Research Center Trustworthy Data Science
and Security (https://rc-trust.ai), one of the Research Alliance centers within the UA Ruhr (https://uaruhr.de).

%% file: appendix.tex

\section{Proof of Theorem~\ref{thm:ratlsat}}
Recall that Theorem~\ref{thm:ratlsat} states that \ratl satisfiability is \exptime-complete.

\begin{proof}[Proof of Theorem~\ref{thm:ratlsat}]
For the upper bound we follow a reasoning similar to the one used for the satisfiability for \rctl \cite{rctl}, opportunely extended to deal with \atl. 
Precisely, \ratl satisfiability can
be solved by translating a given \ratl formula and a given truth value into an equivalent alternating $\mu$-calculus formula (see de~Alfaro and Henzinger~\cite{concomegagames} and de~Alfaro, Henzinger, and Majumdar~\cite{de2001verification}) of linear size and then checking the resulting formula for satisfiability. The procedure to which we refer is used to translate an \atlstar formula $\varphi$ into an equivalent alternating $\mu$-calculus formula $\varphi'$, and relies on the \nrsbc games introduced in Subsection~\ref{subsec:modelchecking}. 

Note that the  translation from \atlstar to alternating $\mu$-calculus is in general exponential in the maximal number of nested temporal operators. More precisely, let $\Sigma =\set{\cNext,U, R,\Diamond,\Box}$ the set of all classical temporal operators used in \atlstar. 
The translation is exponential in the length $i$ of the longest path in the syntax tree labeled by operators in $\Sigma$.
In our specific case of $\varphi$ being an \ratl formula, we have $i=2$ since the interpretations of the temporal operators $\Nextdot, \Diamonddot$, and $\Boxdot$ correspond to a nesting of at most two classical temporal operators in $\Sigma$.
Since the satisfiability problem for the alternating $\mu$-calculus is \exptime-complete \cite{DBLP:journals/jacm/AlurHK02}, \ratl satisfiability is in \exptime.

A matching lower bound already holds for \ctl \cite{ctlsat}, and thus also for \ratl. Aain, the lower bound for \ctl already holds for the fragment without until and release~\cite{meieretal}.
\end{proof}

\section{Proof of Lemma~\ref{lemma:ratl-as-expressive-as-atl}}
\label{app:proof-ratl-as-expressive-as-atl}

In order to prove Lemma~\ref{lemma:ratl-as-expressive-as-atl}, we have to introduce the logic \atl first~\cite{DBLP:journals/jacm/AlurHK02}.
We then show the statement by induction.

\paragraph{The logic \atl}
The syntax of \atl is identical to that of \ratl as introduced in Section~\ref{sec:rATL}, except for ``un-dotted'' temporal operators.
More precisely, its syntax is given by the grammar
\begin{align*}
	\phi & \Coloneqq p \mid \neg \phi \mid \phi \vee \phi \mid \phi \land \phi \mid \stratDiamond{A}\Phi \mid \stratBox{A}\Phi \\
	\Phi & \Coloneqq \Next\phi \mid \Diamond\phi \mid  \Box \phi,
\end{align*}
where $p$ ranges over atomic propositions and $A$ ranges over subsets of agents.
As usual, \atl distinguishes between state formulas (those derivable from $\varphi$) and path formulas (those derivable from $\Phi$).
If not specified otherwise, an \atl formula is a state formula.
    
Note that we have included conjunction and the $\stratBox{~}$-operator explicitly in the syntax.
We do this because we will later only consider \atl formulas in negation normal form.
Furthermore, let us  mention that we do not consider the until-operator since we have defined \ratl without it.

The semantics of ATL is defined in terms of a satisfaction relation $\models$ that relates the states\slash{}paths of a concurrent game structure with all ATL formulas that are satisfied at that state.
In what follows, let $\struct = (\states, \ag, \ac, \tr, \lab)$ be a concurrent game structure, $s$ a state of $\struct$, $\pth$ a path of $\struct$, and $S \subseteq \ag$ a set of agents.
\begin{align*}
    \struct, s \models p & \Leftrightarrow p \in \ell(s) \\
    \struct, s \models \lnot \varphi & \Leftrightarrow \struct, s \not\models \varphi, \\
    \struct, s \models \varphi_1 \lor \varphi_2 & \Leftrightarrow \struct, s \models \varphi_1 \text{ or } \struct, s \models \varphi_2, \\
    \struct, s \models \varphi_1 \land \varphi_2 & \Leftrightarrow \struct, s \models \varphi_1 \text{ and } \struct, s \models \varphi_2, \\
    \struct, s \models \stratDiamond{A} \Phi & \Leftrightarrow \text{there is a set~$\stratset_A$ of strategies, one for} \\
    & \phantom{{}\Leftrightarrow{}} \text{each agent in $A$, such that for all paths} \\
    & \phantom{{}\Leftrightarrow{}} \text{$\pth \in \outcome(s, \stratset_A)$ we have $\struct, \pth \models \Phi$}, \\
    \struct, s \models \stratBox{A} \Phi & \Leftrightarrow \text{for every set~$\stratset_A$ of strategies, one for} \\
    & \phantom{{}\Leftrightarrow{}} \text{each agent in $A$, there exists a path} \\
    & \phantom{{}\Leftrightarrow{}} \text{$\pth \in \outcome(s, \stratset_A)$ such that $\struct, \pth \models \Phi$}. \\
    \intertext{Moreover, we have}
    \struct, \pth \models \Next \varphi & \Leftrightarrow \struct, \pth[1] \models \varphi \\
    \struct, \pth \models \Diamond \varphi & \Leftrightarrow \text{$\struct, \pth[i] \models \varphi$ for an $i \geq 0$}, \\
    \struct, \pth \models \Box \varphi & \Leftrightarrow \text{$\struct, \pth[i] \models \varphi$ for all $i \geq 0$}.
\end{align*}

We are now ready to prove Lemma~\ref{lemma:ratl-as-expressive-as-atl}.

\paragraph{Proof}
Recall the statement of Lemma~\ref{lemma:ratl-as-expressive-as-atl}:
for every \atl formula $\phi$, there exists an \ratl formula $\phi^\star$ such that for every concurrent game structure~$\struct$ and every state~$s$ of $\struct$ we have $\eval(s, \phi^\star) = 1111$ if and only if $\struct, s \models \phi$.

Without loss of generality, we assume in the remainder that the \atl formula~$\varphi$ is in negation normal form (i.e., the only Boolean operators are negation, disjunction, and conjunction, and negation does only appear at atomic propositions).
We make this assumption because the semantics of the implication does not easily map to the semantics of its non-robust counterpart.
Moreover, pushing negations to the atomic propositions slightly simplifies the proof of Lemma~\ref{lemma:ratl-as-expressive-as-atl}.
It is worth noting, however, that this is a minor deviation from the literature (e.g., Nayak et al.~\cite{rctl} only removed implications but did not go to negation normal form).

If $\varphi$ is not in negation normal form, we can easily transform it as mentioned in the proof sketch of Lemma~\ref{lemma:ratl-as-expressive-as-atl} in Subsection~\ref{subsec_expressiveness}:
\begin{enumerate}
    \item We eliminate every implication $\phi_1 \rightarrow \phi_2$ in $\varphi$ by replacing it with the expression~$\lnot \phi_1 \lor \phi_2$.
    \item We push all negations inwards to the level of atomic propositions.
\end{enumerate}
It is not hard to verify that these two steps indeed transform any \atl formula into an equivalent \atl formula in negation normal form.

To construct the \ratl formula $\varphi^\star$, we apply the following recursive transformation $t$, which essentially just ``dots'' the temporal operators.
For state formulas, we define
\begin{align*}
    t(p) & \coloneqq p, \\
    t(\lnot p) & \coloneqq \lnot p, \\
    t(\varphi_1 \lor \varphi_2) & \coloneqq t(\varphi_1) \lor t(\varphi_2), \\
    t(\varphi_1 \land \varphi_2) & \coloneqq t(\varphi_1) \land t(\varphi_2), \\
    t(\stratDiamond{A} \Phi) & \coloneqq \stratDiamonddot{A} t(\Phi), \\
    t(\stratBox{A} \Phi) & \coloneqq \stratBoxdot{A} t(\Phi). \\
    \intertext{For path formulas, we define}
    t(\Next \varphi) & \coloneqq \Nextdot t(\varphi), \\
    t(\Diamond \varphi) & \coloneqq \Diamonddot t(\varphi), \\
    t(\Box \varphi) & \coloneqq \Boxdot t(\varphi). \\
\end{align*}
Given an \atl formula $\varphi$ in negation normal form, we then simply set $\varphi^\star \coloneqq t(\varphi)$.
Note that $\varphi$ and $\varphi^\star$ are structurally identical, except for the ``dotting'' of temporal operators.

To prove Lemma~\ref{lemma:ratl-as-expressive-as-atl}, we show the following, slightly more technical statement.

\begin{lemma} \label{lem:slightly-stronger-than-lem2}
For every \atl state formula $\phi$ in negation normal form, every \atl path formula $\Phi$ in negation normal form, every concurrent game structure~$\struct$, every state~$s$ of $\struct$, and every path $\pth$ of $\struct$, the transformed formulas $t(\varphi)$ and $t(\Phi)$ satisfy
\begin{enumerate}
    \item \label{itm:proof:lemma:ratl-as-expressive-as-atl:1}
    $\eval(s, t(\phi)) = 1111$ if and only if $\struct, s \models \phi$; and
    \item \label{itm:proof:lemma:ratl-as-expressive-as-atl:2}
    $\eval(\pth, t(\Phi)) = 1111$ if and only if $\struct, \pth \models \Phi$.
\end{enumerate}
\end{lemma}

\begin{proof}[Proof of Lemma~\ref{lem:slightly-stronger-than-lem2}]
We prove the statement by induction over the structure of $\varphi$.

For the base case, let $\varphi = p$.
We first observe that $t(\varphi) = p$ as well.
Hence, we have $\struct, s \models p$ and $\eval(s, p) = 1111$ if $p \in \ell(s)$ by the definitions of the \atl and \ratl semantics (recall that $\ell$ is the  labeling function of the concurrent game structure).
Conversely, $\struct, s \not\models p$ and $\eval(s, p) = 0000$ if $p \notin \ell(s)$.
Thus, the claim holds.

For the induction step, we make a case distinction:
\begin{itemize}
    \item Let $\varphi = \lnot \varphi_1$.
    Since $\varphi$ is in negation normal form, we know that negation can only appear at the level of atomic propositions (i.e., $\varphi_1 = p$).
    Therefore, $\varphi = \lnot p$ and $t(\varphi) = \lnot p$.
    
    We now distinguish between $p \in \ell(s)$ and $p \notin \ell(s)$.
    If $p \not\in \ell(s)$, then $\eval(s, p) = 0000$ and $\eval(s, \lnot p) = 1111$ by definition of the \ratl semantics.
    Similarly, we have $\struct, s \not\models p$ and $\struct, s \models \lnot p$ by definition of the \atl semantics.

    On the other hand, if $p \in \ell(s)$, then $\eval(s, p) = 1111$ and $\eval(s, \lnot p) = 0000$.
    Similarly, we have $\struct, s \models p$ and $\struct, s \not\models \lnot p$.
    \item Let $\varphi = \varphi_1 \lor \varphi_2$ and, hence, $t(\varphi) = t(\varphi_1) \lor t(\varphi_2)$.
    By induction hypothesis, we know that $\eval(s, t(\varphi_1)) = 1111$ if and only if $\struct, s \models \varphi_1$ and $\eval(s, t(\varphi_2)) = 1111$ if and only if $\struct, s \models \varphi_2$.
    
    We now distinguish the two cases $\eval(s, t(\varphi)) = 1111$ and $\eval(s, t(\varphi)) \neq 1111$.
    If $\eval(s, t(\varphi)) = 1111$, then $\eval(s, t(\varphi_1)) = 1111$ or $\eval(s, t(\varphi_2)) = 1111$ because $\eval(s, t(\varphi)) = \max{\bigl\{ \eval(s, t(\varphi_1)), \eval(s, t(\varphi_2)) \bigr\}}$ by definition of the \ratl semantics.
    Applying the induction hypothesis then yields $\struct, s \models \varphi_1$ or $\struct, s \models \varphi_2$.
    Thus, $\struct, s \models \varphi$.
    
    Similarly, if $\eval(s, t(\varphi)) \neq 1111$, then $\eval(s, t(\varphi)) \prec 1111$ because $1111$ is the largest truth value in $\boolfour$.
    Thus, $\eval(s, t(\varphi_1)) \prec 1111$ and $\eval(s, t(\varphi_2)) \prec 1111$ by definition of the \ratl semantics, showing that $\eval(s, t(\varphi_1)) \neq 1111$ and $\eval(s, t(\varphi_2)) \neq 1111$.
    Applying the induction hypothesis then yields $\struct, s \not\models \varphi_1$ and $\struct, s \not\models \varphi_2$, implying $\struct, s \not\models \varphi$ since $\varphi = \varphi_1 \lor \varphi_2$.
    \item Let $\varphi = \varphi_1 \land \varphi_2$ and, hence, $t(\varphi) = t(\varphi_1) \land t(\varphi_2)$.
    Since this case in analogous to $\varphi = \varphi_1 \land \varphi_2$ (with $\min$ used for $\max$), we skip it here.
    \item For a set $A \subseteq \ag$ of agents, let $\varphi = \stratDiamond{A} \Phi$ and, hence, $t(\varphi) = \stratDiamonddot{A} t(\Phi)$.
    
    If $\eval(s, t(\varphi)) = 1111$, then there exists a set~$\stratset_A$ of strategies, one for each agent in $A$, such that for all paths~$\pth \in \outcome(s, \stratset_A)$ we have $\eval(\pth, t(\Phi)) \succeq 1111$.
    In particular, this means that $\eval(\pth, t(\Phi)) = 1111$ holds for all such paths~$\pi$ since $1111$ is the largest truth value in $\boolfour$.
    By applying the induction hypothesis, we also know that the strategies in $\stratset_A$ ensure $\struct, \pth \models \Phi$ for all paths~$\pth \in \outcome(s, \stratset_A)$.
    Thus, $\struct, s \models \varphi$ since $\varphi = \stratDiamond{A} \Phi$.
    
    On the other hand, if $\eval(s, t(\varphi)) \neq 1111$, then $\eval(s, t(\varphi)) \prec 1111$.
    This means that for every set $\stratset_A$ of strategies, there exists a path~$\pth \in \outcome(s, \stratset_A)$ with $\eval(\pth, t(\Phi)) \prec 1111$ (in particular, $\eval(\pth, t(\Phi)) \neq 1111$).
    Applying the induction hypothesis then yields that for every set $\stratset_A$ of strategies, there exists a path~$\pth \in \outcome(s, \stratset_A)$ with $\struct, \pth \not\models \Phi$.
    Hence, $\struct, s \not\models \varphi$.
    \item For a set $A \subseteq \ag$ of agents, let $\varphi = \stratBox{A} \Phi$ and, hence, $t(\varphi) = \stratBoxdot{A} t(\Phi)$.
    Since this case in analogous to $\varphi = \stratDiamond{A} \Phi$ (with the case distinction for $\eval(s, t(\varphi)) = 1111$ and $\eval(s, t(\varphi)) \neq 1111$ swapped), we skip it here.
    \item Let $\varphi = \Next \varphi_1$ and, hence, $t(\varphi) = \Nextdot t(\varphi_1)$.
    
    Since $\eval(\pth, \Nextdot t(\phi_1)) = b_1 b_2 b_3 b_4$ with $ b_k = \eval(\pth[1], t(\varphi))[k]$ by definition of the \ratl semantics, we have $\eval(\pth, \Nextdot t(\varphi_1)) = 1111$ if and only if $\eval(\pth[1], t(\varphi_1)) = 1111$.
    By induction hypothesis, this is equivalent to $\struct, \pth[1] \models \varphi_1$ and, by definition of the \atl semantics, $\struct, \pth \models \varphi$ since $\varphi = \Next \varphi_1$.
    \item Let $\varphi = \Diamond \varphi_1$ and, hence, $t(\varphi) = \Diamonddot t(\varphi_1)$.
    
    Since $\eval(\pth, \Diamonddot t(\varphi_1)) = b_1b_2b_3b_4$ with $b_k = \max_{i \ge 0} \eval(\pth[i], t(\varphi_1))[k]$ by definition of the \ratl semantics, we know that $\eval(\pth, \Diamonddot t(\varphi_1)) = 1111$ if and only if there exists an $i \geq 0$ with $\eval(\pth[i], t(\varphi_1)) = 1111$.
    By induction hypothesis, this is equivalent to $\struct, \pth[i] \models \varphi_1$ and, by definition of the \atl semantics, $\struct, \pth \models \varphi$ since $\varphi = \Diamond \varphi_1$.
    \item Let $\varphi = \Box \varphi_1$ and, hence, $t(\varphi) = \Boxdot t(\varphi_1)$.
    
    By definition of the \ratl semantics, we know that $\eval(\pth, \Boxdot t(\varphi_1))$ is a truth value $b_1 b_2 b_3 b_4 \in \boolfour$ with $b_1 = \min_{i \ge 0} \eval(\pth[i], t(\phi_1))[1]$.
    Thus, $\eval(\pth, \Boxdot t(\varphi_1)) = 1$ can only hold if $\eval(\pth[i], t(\phi_1))[1] = 1$ holds for all $i \geq 0$.
    Since $1111$ is the largest truth value, this is equivalent to the statement that $\eval(\pth[i], t(\phi_1)) = 1111$ holds for all $i \geq 0$.
    Moreover, by applying the induction hypothesis, we obtain that $\struct, \pth[i] \models \varphi_1$ holds for all $i \geq 0$.
    Thus, $\struct, \pth \models \varphi$ since $\varphi = \Box \varphi_1$.
\end{itemize}
This concludes the proof.
\end{proof}

Lemma~\ref{lemma:ratl-as-expressive-as-atl} now follows immediately from Item~\ref{itm:proof:lemma:ratl-as-expressive-as-atl:1} of Lemma~~\ref{lem:slightly-stronger-than-lem2}, the fact that every \atl formula can be transformed into negation normal form, and the fact that $\varphi^\star = t(\varphi)$.

\section{Definition of rCTL}
\label{app:rCTL}
For the reader's convenience, we here repeat the definition of rCTL as introduced by Nayak et al.~\cite{rctl}.

\paragraph{Syntax of \rctl}
Formulas of \rctl are classified into state and path formulas.
\rctl state formulas are formed according to the grammar
\[ \varphi \Coloneqq p \mid \neg \varphi \mid \varphi \vee \varphi \mid \varphi \wedge \varphi \mid \varphi \rightarrow \varphi \mid \exists \Phi \mid \forall \Phi,\]
where $p\in \ap$ is an atomic proposition and $\Phi$ is a path formula.
On the other hand, rCTL path formulas are derived from the grammar
\[ \Phi \Coloneqq \Nextdot \varphi \mid \Diamonddot \varphi \mid \Boxdot \varphi.\]
If not stated otherwise, an \rctl formula is a state formula.
Again, we omit the until-operator for the sake of simplicity since our definition of \ratl does not contain it.

\paragraph{Kripke structures}
An rCTL formula is evaluated on a Kripke structure, a mathematical object modeling a system.
Following the notation of Nayak et al.~\cite{rctl}, a Kripke structure over a set $\ap$ of atomic propositions is a tuple $\trans = (S, I, R, L)$ where $S$ is a finite set of states, $I\subseteq S$ is the set of initial states, $R\subseteq S\times S$ is the a transition relation, and $L\colon S\rightarrow 2^{\ap}$ is the labeling function.
Without loss of generality, we assume that for all states $s \in S$, there exists a state $s'$ satisfying $(s,s')\in R$.

A path of the Kripke structure $\trans$ is an infinite sequence of states $\pi = s_0s_1\cdots$ such that $(s_i, s_{i+1}) \in R$ for every $i \geq 0$.
Moreover, let $\mathit{paths}(s)$ denote the set of all paths starting from state $s \in S$.
Finally, for a path $\pi$ and $i\geq 0$, we use $\pi[i]$ to denote the $i$-th state of $\pi$ and $\pi[i..]$ to denote the suffix of $\pi$ starting at index $i$.

\paragraph{Semantics of rCTL}
The rCTL semantics is a mapping $V_\mathit{CTL}$ that assigns an element of $\boolfour$ to every pair of state and state formula and to every pair of path and path formula.
Let us begin with state formulas, where $s \in S$ is a state and $p \in AP$ is an atomic proposition:
\begin{align*}
	 V_{\rctl}(s,p) & = \begin{cases}
        0000 & \text{if $p\not \in L(s)$; and} \\
        1111 & \text{if $p\in L(s)$,}
	\end{cases} \\
    V_{\rctl}(s,\neg \Phi) & = \begin{cases}
        0000 & \text{if $V_{\rctl}(s,\phi) = 1111$; and} \\
        1111 & \text{if $V_{\rctl}(s, \phi) \prec 1111$,}
    \end{cases}\\
    V_{\rctl}(s,\Phi \vee \Psi) & = \max\{V_{\rctl}(s,\Phi), V_{\rctl}(s,\Psi)\}, \\
    V_{\rctl}(s,\Phi \wedge \Psi) & = \min\{V_{\rctl}(s,\Phi), V_{\rctl}(s,\Psi)\}, \\
    V_{\rctl}(s,\Phi \rightarrow \Psi) & = 
    V_{\rctl}(s,\Phi) \Rightarrow V_{\rctl}(s,\Psi).
    \intertext{where, for two truth values $a, b \in \boolfour$, we have}
    a\Rightarrow b & = \begin{cases}
        1111 &\text{ if $a\preceq b$; and}\\
        b &\text{ if $a \succ b$.}
    \end{cases}
\end{align*}

For the existential and universal path quantification, the definition of $V_{\rctl}$ is as follows:
\begin{align*}
    V_{\rctl}(s,\exists \varphi) & = \max_{\pi \in \mathit{paths}(s)} V_{\rctl}(\pi,\varphi),\\
    V_{\rctl}(s,\forall \varphi) & = \min_{\pi \in \mathit{paths}(s)} V_{\rctl}(\pi,\varphi).
\end{align*}

Finally, the semantics for path formulas is as follows:
\begin{align*}
    V_{\rctl}(\pi,\Nextdot \Phi) & = V_{\rctl}(\pi[1], \Phi),\\
    %
    %
    V_{\rctl}(\pi,\Diamonddot \Phi) & = \max_{i\geq 0} V_{\rctl}(\pi[i],\Phi),\\
    V(\pi,\Boxdot \Phi) & = b_1 b_2 b_3 b_4, \\
    \intertext{where}
    b_1 & = \min_{i\geq 0} V_{\rctl}(\pi[i], \varphi)[1],\\
    b_2 & = \max_{j\geq 0} \min_{i\geq j} V_{\rctl}(\pi[i], \varphi)[2],\\
    b_3 & = \min_{j\geq 0} \max_{i\geq j} V_{\rctl}(\pi[i], \varphi)[3],\\
    b_4 & = \max_{i\geq 0} V_{\rctl}(\pi[i], \varphi)[4].
\end{align*}
We omit the semantics of the Boolean operators as it is same we have given for rATL.
Recall that for a truth value $b = b_1 b_2 b_3 b_4 \in \boolfour$ and $i \in \{ 1, \ldots, 4\}$, we use $b[i]$ to denote the bit $b_i$,
as for rATL.

\section{Proof of Lemma~\ref{lemma:ratl-as-expressive-as-rctl}}
\label{ratl-as-expressive-as-rctl}

Recall the statement of Lemma~\ref{lemma:ratl-as-expressive-as-rctl}:
for every \rctl formula~$\varphi$, there exists an \ratl formula $\phi^\star$ such that for every Kripke structure~$\trans$ and every state~$s$ of $\trans$ we have $\eval(s, \phi^\star) = \eval_{\rctl}(s, \phi)$.

We proceed with the proof of Lemma~\ref{lemma:ratl-as-expressive-as-rctl} in three steps, as described in Subsection~\ref{subsec_expressiveness}.

\paragraph{Step 1}
First, we transform the Kripke structure $\trans = (S, I, R, L)$ into an ``equivalent'' concurrent game structure $\struct_\trans$ with a single-agent $a$.
This transformation is required because \rctl formulas are evaluated over Kripke structures, whereas \ratl formulas are evaluated over concurrent game structures.
Formally, we define $\struct_\trans = (\states, \ag, \ac, \tr, \lab)$ with
\begin{itemize}
    \item $\states = S$;
    \item $\ag = \{ a \}$;
    \item $\ac = S$;
    \item $\delta(s, s') = \begin{cases} s' & \text{if $(s, s') \in R$}; \\ s'' & \text{for some $(s, s'') \in R$ if $(s, s') \notin R$};  \end{cases}$
    \item $\ell(s) = L(s)$ for each $s \in S$.
\end{itemize}
Note that actions in $\struct_\trans$ are the states of the Kripke structure, and $\delta$ is well defined because $\trans$ has no dead ends.
Moreover, the transition function $\tr$ mimics the transition relation $R$, except that it (potentially) contains additional parallel edges to make the function complete.
Overall, it is not hard to verify that $\trans$ and $\struct_\trans$ have the same set of paths.

\paragraph{Step 2}
Second, given an \rctl formula $\varphi$, we construct the \ratl formula $\varphi^\star$.
To this end, we use a mapping $t'$ that replaces the path quantifiers $\exists$ and $\forall$ with the corresponding strategy quantifiers $\stratDiamonddot{\{a\}}$ and $\stratDiamond{\emptyset}$, respectively.
For state formulas, we define
\begin{align*}
    t'(p) & \coloneqq p, \\
    t'(\lnot \varphi) & \coloneqq \lnot t'(\varphi), \\
    t'(\varphi_1 \lor \varphi_2) & \coloneqq t'(\varphi_1) \lor t'(\varphi_2) \\
    t'(\varphi_1 \land \varphi_2) & \coloneqq t'(\varphi_1) \land t'(\varphi_2) \\
    t'(\varphi_1 \rightarrow \varphi_2) & \coloneqq t'(\varphi_1) \rightarrow t'(\varphi_2), \\
    t'(\exists \Phi) & \coloneqq \stratDiamonddot{\{a\}} t'(\Phi), \\
    t'(\forall \Phi) & \coloneqq \stratDiamonddot{\emptyset}t'( \Phi). \\
    \intertext{For path formulas, we define}
    t'(\Nextdot \varphi) & \coloneqq \Nextdot t'(\varphi), \\
    t'(\Diamonddot \varphi) & \coloneqq \Diamonddot t'(\varphi), \\
    t'(\Boxdot \varphi) & \coloneqq \Boxdot t'(\varphi).
\end{align*}
Given an \rctl formula $\varphi$, we then simply set $\varphi^\star \coloneqq t'(\varphi)$.

\paragraph{Step 3}
It is left to show that $\eval(s, \varphi^\star) = \eval_{\rctl}(s, \varphi)$ holds for all states $s$.
To this end, we show the following, slightly stronger statement.

\begin{lemma} \label{lem:slightly-stronger-than-lem3}
For every \rctl state formula $\varphi$, every \rctl path formula $\Phi$, every Kripke structure~$\trans$, every state~$s$ of $\trans$, and every path $\pth$ of $\trans$, the transformed formulas $t'(\varphi)$ and $t'(\Phi)$ satisfy the following in the concurrent game structure $\struct_\trans$:
\begin{enumerate}
    \item \label{itm:proof:lem3:1}
    $\eval(s, t'(\varphi)) = \eval_{\rctl}(s, \varphi)$; and
    \item \label{itm:proof:lem3:2}
    $\eval(\pth, t'(\Phi)) = \eval_{\rctl}(\pth, \Phi)$.
\end{enumerate}
\end{lemma}

\begin{proof}[Proof of Lemma~\ref{lem:slightly-stronger-than-lem3}]
The base case (i.e., atomic propositions) and the induction step for the Boolean operators ($\lnot, \lor, \land, \rightarrow$), and the robust temporal operators ($\Nextdot, \Diamonddot, \Boxdot$) follow from applying the semantics of \rctl and \ratl as expected.
Hence, we skip them here and only investigate the remaining two cases for the existential and universal path quantifiers:
\begin{itemize}
    \item Let $\varphi = \exists \Phi$ and, therefore, $t'(\varphi) = \stratDiamond{\{a\}} t'(\Phi)$.

    Recall that the paths of $\trans$ and $\struct_\trans$ are identical and $\struct_\trans$ is a concurrent game structure with a single-agent.
    Thus, for every path $\pth \in \mathit{paths(s)}$, there exists a strategy $f_\pth$ such that $\outcome(s, \{ f_\pth \}) = \{ \pth \}$.
    Conversely, every strategy $f_a$ of agent~$a$ produces exactly one path $\pth_{f_a} \in \mathit{paths(s)}$.

    Next, we apply the induction hypothesis and obtain $\eval(\pth, t'(\Phi)) = \eval_{\rctl}(\pth, \Phi)$ for all paths $\pth \in \mathit{paths}(s)$.
    Thus,
    \[ \max_{f_a \in F_a}{\eval(\pth_{f_a}, t'(\Phi))} = b = \max_{\pth \in \mathit{paths}(s)}{\eval_{\rctl}(\pth, \Phi)}, \]
    where $F_a$ is the set of all strategies of agent~$a$ starting in state~$s$.
    This observation implies $\eval(s, \stratDiamonddot{\{ a \}} t'(\Phi)) = b$ by definition of the $\stratDiamonddot{~}$-operator (recall that $\outcome(s, f_a)$ is a singleton set of all strategies $f_a \in F_a$).
    On the other hand, we have $\eval_{\rctl}(s, \exists \Phi) = b$ by definition of the $\exists$-operator,
    In total, we obtain $\eval(s, \stratDiamonddot{\{ a \}} t'(\Phi)) = \eval_{\rctl}(s, \exists \Phi)$, proving the claim.

    \item Let $\varphi = \forall \Phi$ and, therefore, $t'(\varphi) = \stratDiamond{\emptyset} t'(\Phi)$.
    
    Again, recall now that the paths of $\trans$ and $\struct_\trans$ are identical.
    Moreover, note that $\outcome(s, \emptyset) = \mathit{paths}(s)$ because $\struct$ is a concurrent game structure with a single agent.

    Next, we apply the induction hypothesis and obtain $\eval(\pth, t'(\Phi)) = \eval_{\rctl}(\pth, \Phi)$ for all paths $\pth \in \mathit{paths}(s)$, as before.
    Thus,
    \[ \min_{\pth \in \outcome(s, \emptyset)}{\eval(\pth, t'(\Phi))} = b = \min_{\pth \in \mathit{paths}(s)}{\eval_{\rctl}(\pth, \Phi)}. \]
    
    This observation implies $\eval(s, \stratDiamonddot{\emptyset} t'(\Phi)) = b$ by definition of the $\stratDiamonddot{~}$-operator (recall that $\outcome(s, f_a) = \mathit{paths}(s)$).
    On the other hand, we have $\eval_{\rctl}(s, \forall \Phi) = b$ by definition of the $\exists$-operator,
    In total, we obtain $\eval(s, \stratDiamonddot{\emptyset} t'(\Phi)) = \eval_{\rctl}(s, \forall \Phi)$, proving the claim.

    
\end{itemize}
This concludes the proof.
\end{proof}

Lemma~\ref{lemma:ratl-as-expressive-as-rctl} now follows immediately from Item~\ref{itm:proof:lem3:1} of Lemma~\ref{lem:slightly-stronger-than-lem3} and the fact $\varphi^\star = t'(\varphi)$.

\section{Proof of Lemma~\ref{lemma:ratlstar2atlstar}}

Recall that we want to prove that for every \ratlstar formula~$\phi$ and every truth value~$t \in \boolfour$, there is an \atlstar formula~$\phi_t$ such that $\eval(s, \phi) \succeq t$ if and only if $\struct, s \models \phi_t$.
Furthermore, we will show function mapping $\phi$ and $t$ to $\phi_t$ is polynomial-time computable.

\begin{proof}[Proof of Lemma~\ref{lemma:ratlstar2atlstar}]
Before we show the inductive construction of $\phi_t$, let us remark that we can restrict ourselves to $t \succ 0000$, as $\eval(s, \phi) \succeq 0000$ is true for every state~$s$ and every formula~$\phi$. Hence, we can pick $\phi_{0000} = p \vee \neg p$ for some atomic proposition~$p$.

In the following, we assume $t \succ 0000$ and define
\begin{itemize}
    \item $p_t = p$,
    \item $(\neg \phi)_t = \neg \phi_t$,
    \item $(\phi_1 \vee \phi_2)_t = (\phi_1)_t \vee (\phi_2)_t$,
    \item $(\phi_1 \wedge \phi_2)_t = (\phi_1)_t \wedge (\phi_2)_t$,
    \item $(\phi_1 \rightarrow \phi_2)_{1111} = \bigwedge_{t \succeq 0000} (\phi_2)_t \vee \neg (\phi_1)_t$ and 
    \item $(\phi_1 \rightarrow \phi_2)_{t} = (\phi_1 \rightarrow \phi_2)_{1111} \vee (\phi_2)_t$ for $t \prec 1111$,
    \item $(\stratDiamonddot{A}\Phi)_t = \stratDiamond{A}\Phi_t$, and
    \item $(\stratBoxdot{A}\Phi)_t = \stratBox{A}\Phi_t$.
\end{itemize}
For Boolean combinations of path formulas, the translation is defined analogously as for state formulas while for temporal operators, the translation is defined as 
\begin{itemize}
    \item $(\Nextdot\Phi)_t = \Next\Phi_t$,
    \item $(\Diamonddot\Phi)_t = \Diamond\Phi_t$, and
    \item $(\Boxdot\Phi)_{1111} = \Box\Phi_{1111}$,
    \item $(\Boxdot\Phi)_{0111} = \Diamond\Box\Phi_{0111}$,
    \item $(\Boxdot\Phi)_{0011} = \Box\Diamond\Phi_{0011}$, and
    \item $(\Boxdot\Phi)_{0001} = \Diamond\Phi_{0001}$.
\end{itemize}
An induction over the construction of $\phi$ shows that $\phi_t$ has the desired properties.
\end{proof}

\section{Proof of Theorem~\ref{thm:ratlstarresults}}

\begin{proof}[Proof of Theorem~\ref{thm:ratlstarresults}]
The upper bounds follow directly from Lemma~\ref{lemma:ratlstar2atlstar} and that fact that \atlstar model-checking and satisfiability are in \twoexptime~\cite{DBLP:journals/jacm/AlurHK02,Schewe}, while the lower bounds follow from the fact that \atlstar (and thus \ctlstar) is a fragment of \ratlstar:
Generalizing the proof of Lemma~\ref{lemma:ratl-as-expressive-as-atl}, given an \atlstar formula~$\phi$, we eliminate all implications, push all negations to the atomic propositions, and then dot every operator, thereby obtaining an \ratlstar formula~$\phi^\star$ such that $
\struct, s\models \phi \text{ if and only if } \eval(s,\phi^\star) = 1111$.

Now, as the \ctlstar satisfiability problem is \twoexptime-complete (even for the fragment without until and release~\cite{meieretal}), the \ratlstar satisfiability problem is \twoexptime-hard. 
To obtain \twoexptime-hardness of the model-checking problem we rely on the reduction from the \ltl realizability problem to the \atlstar model-checking problem~\cite{DBLP:journals/jacm/AlurHK02} and the fact that \ltl realizability is \twoexptime-hard, even for the fragment without until and release~\cite{boxesanddiamonds}.
\end{proof}